\pdfoutput=1
\documentclass[12pt]{article}
\usepackage[a4paper, top=2.5cm, bottom=2.5cm, left=2.2cm, right=2.2cm]
{geometry}
\usepackage{parskip}

\usepackage{mathtools,amssymb,amsthm,bm}
\usepackage{physics}

\usepackage{graphicx}
\usepackage{caption, subcaption}
\usepackage{float}
\usepackage[dvipsnames]{xcolor}
\usepackage{tikz}
\usepackage{tikz-cd}
\usepackage[most]{tcolorbox}

\tikzset{
    plumbing vertex/.style={
        circle,
        fill,
        draw,
        inner sep=1.5pt
    }
}

\usepackage{array}
\usepackage{enumitem}
\usepackage{comment}
\usepackage{authblk} 
\usepackage[titletoc,title]{appendix}

\usepackage{microtype}

\usepackage[backend=biber, style=numeric-comp, sorting=none]{biblatex}
\usepackage{hyperref}
\hypersetup{hidelinks}

\newtheorem{theorem}{Theorem}

\newtheorem{definition}{Definition}

\newcommand{\Z}{{\mathbb Z}}
\newcommand{\R}{{\mathbb R}}
\newcommand{\C}{{\mathbb C}}
\newcommand{\Q}{{\mathbb Q}}

\def\tilde{\widetilde}
\renewcommand{\hat}{\widehat}

\def\be{\begin{equation}}
\def\ee{\end{equation}}
\def\bea{\begin{align}}
\def\eea{\end{align}}

\title{Mock Modularity, Resurgence, and \\ Dual False Theta Functions}
\author{Mrunmay Jagadale}
\date{\vspace{-0.6cm}}
\affil{California Institute of Technology, Pasadena, CA 91125, USA}

\begin{document}

\maketitle

\begin{abstract}
\noindent The $\hat{Z}$ invariant of a three-manifold $M_{3}$, a half-index of the $3d$ $\mathcal{N}=2$ theory $T[M_{3}]$, is well understood for negative-definite plumbed three-manifolds. For orientation-reversed manifolds, however, the known formulas fail to produce $q$-series, and making sense of the operation $q \rightarrow q^{-1}$ remains a central open problem. Seifert manifolds, whose $\hat{Z}$ invariants are built from false theta functions, offer a natural testing ground. We study ``dual false theta functions'', the objects that should replace false theta functions under orientation reversal. Two characterizations of these duals exist in the literature: a modular one, based on mock modularity with a prescribed shadow, and a resurgent one, based on Borel--Mordell integrals and transseries. We prove that the two characterizations are equivalent. We construct a new family of dual false theta functions for values of positive integer $p$ satisfying a Pell-equation condition. Its members are built from Zwegers' indefinite theta functions, have integer coefficients, and have effective central charge $c_{\text{eff}} \leq 1$. We also show that the mock theta functions of Li and Schwagenscheidt satisfy these characterizations and arise from a natural regularization of the divergent product $\eta^{3}(\tau)\psi_{p, a}(- \tau)$.
\end{abstract}

\tableofcontents

%%%%%%%%%%%%%%%%%%%%%%%%%%%%%%%%%%%%%%%

\section{Introduction}

The $3d$-$3d$ correspondence has revealed a rich interplay between three-manifold topology, supersymmetric quantum field theory, and number theory. To a three-manifold $M_{3}$, it associates a $3d$ $\mathcal{N}=2$ theory $T[M_{3}]$, obtained by compactifying the $6d$ $(2,0)$ theory on $M_{3}$ \cite{Dimofte:2011ju}. Supersymmetry-protected quantities of $T[M_3]$ then give rise to rich topological invariants of $M_3$. 

A particularly interesting example is the $q$-series-valued invariant $\hat{Z}$ \cite{Gukov:2017kmk, Gukov:2016gkn}. It is defined as a supersymmetric partition function of $T[M_{3}]$ on $S^{1}\times_{q}D^{2}$ with suitable boundary conditions \cite{Gadde:2013wq}. It has connections to diverse areas of mathematics and physics, including resurgence, quantum topology, and quantum modularity \cite{Gukov:2016njj, Cheng:2018vpl}. 

In \cite{Gukov:2017kmk}, Gukov, Pei, Putrov, and Vafa gave a formula for the $\hat{Z}$ invariant of negative-definite plumbed manifolds. For a plumbed manifold with opposite orientation, that is, for a positive-definite plumbed manifold, this formula does not yield a $q$-series with exponents bounded from below. From the physical definition, however, we expect $\hat{Z}$ to be well-defined for manifolds of both orientations. 

The $3d$-$3d$ correspondence relates the $\hat{Z}(q)$ invariant to the partition function of $\mathrm{SL}(2,\C)$ Chern--Simons theory, where $q$ is related to the Chern--Simons level $k$ by $q= e^{\frac{2 \pi i}{k}}$. Under orientation reversal, the sign of the Chern--Simons action is reversed, naturally leading to $q \rightarrow q^{-1}$. However, under $q \rightarrow q^{-1}$, a $q$-series with exponents bounded from below is mapped to a $q^{-1}$-series with exponents bounded from above. Therefore, we need additional structure to reinterpret the result as a $q$-series with exponents bounded from below. Understanding this $q$ to $q^{-1}$ operation has been the subject of considerable recent work \cite{Cheng:2019uzc, Park:2021ufu, Costin:2023kla, Cheng:2024vou,Harichurn:2025suf, Adams:2025aad, Adams:2025qgj, Fantini:2026ale, Fantini:2025wap}.

A particularly useful setting in which to study this problem is provided by Seifert manifolds. Their $\hat{Z}$ invariants are linear combinations of false theta functions \cite{Gukov:2019mnk}, which therefore furnish a natural class of basic examples on which to understand the $q \rightarrow q^{-1}$ operation. We denote the false theta functions by $\psi_{p, a}$, with $p \in \Z_{>0}$ and $a \in \Z_{2p}$.

Following \cite{Costin:2023kla}, we refer to the $q$-series obtained by making sense of the $q \rightarrow q^{-1}$ transformation on a false theta function as a dual false theta function and denote it by $\psi_{p,a}^{\vee}$. There are two seemingly different approaches to constructing such duals in the literature, both of which we will discuss in this paper. The first is based on the quantum modular properties of false theta functions, in which the dual false theta function is characterized by requiring appropriate modular transformation properties \cite{Cheng:2019uzc}. The second is based on the resurgent structure of false theta functions and determines the dual by solving a set of modular bootstrap-like equations \cite{Costin:2023kla}. In Theorem \ref{equivalenceofbridges}, we show that the two approaches are equivalent. In a closely related work \cite{Costin:2026new}, Costin, Dunne, and Saraeb develop a complementary perspective. They study mock theta functions through a resurgent transseries approach built on Mordell--Appell integrals. They also show that two constructions for crossing the natural boundary $|q|=1$, one via modular functional equations and one via Ecalle--Borel summation after $\tau \rightarrow- \tau$, coincide. This is in the same spirit as the equivalence of the modular and resurgent bridges established in Theorem \ref{equivalenceofbridges}.

In \cite{li2022mock}, Li and Schwagenscheidt construct a mock theta function whose modular transformation properties agree with those required of the dual false theta function. Thus, from the perspective of quantum modularity, their mock theta function furnishes a natural candidate for the dual. In Section \ref{dualfalseformulas}, we show that the mock theta functions constructed by Li and Schwagenscheidt can be obtained by regularizing the product $\psi_{p, a}(-\tau) \eta^{3}(\tau)$ and then dividing by $\eta^{3}(\tau)$, where $\tau$ is given by $q = e^{2 \pi i \tau}$. 

In the same Section \ref{dualfalseformulas}, we construct a dual false theta function $\psi^{\vee, \eta}$ for certain values of $p$ and show that it satisfies the modular characterization. It is given by the following simple formula,
$$\frac{-1}{\eta(\tau)} \sum_{\substack{x \in a + 2 p \Z \\ y \in 1+2 \Z \\ 0 < \frac{x}{y} < \frac{p v}{u} }} \mathrm{Sign}(x) \left( \frac{12}{y} \right) q^{\frac{y^{2}}{24 } - \frac{x^{2}  }{4p} }. $$
See the description of $u$, $v$  and $\left( \frac{12}{y} \right)$ in the discussion above \eqref{sparthdualfalsedef}. In Subsection \ref{examplesofdualfalse}, we give examples of dual false theta functions from the two constructions above and compare them with the $q$-series appearing in the literature. In Subsection \ref{nonuniquenessdiscussio}, we discuss the non-uniqueness of dual false theta functions and possible criteria to single out a unique dual. 

%%%%%%%%%%%%%%%%%%%%%%%%%%%%%%%%%%%%%%%

\section{Bridges to the Other Side}\label{bridges}
In this section, we summarize and review two approaches to making sense of the $q \rightarrow q^{-1}$ operation on $q$-series and the background required. 

%%%%%%%%%%%%%%%%%%%%%%%%%%%%%%%%%%%%%%%

\subsection{The Modular Bridge: Mock Modularity and Weil Representation}\label{modularitybridge}

The modular bridge characterizes the dual false theta functions by the modular properties they should satisfy. We begin by reviewing the necessary background: Weil representations, unary theta functions, false theta functions, and mock modular forms. 

Let $\mathrm{Mp}(2,\Z)$ be the metaplectic two-fold cover of $\mathrm{SL}(2,\Z)$ whose elements are pairs $(\gamma, v)$ where $\gamma = \begin{pmatrix}
    a & b \\ c & d
\end{pmatrix} \in \mathrm{SL}(2,\Z)$, and $v : \mathbb{H} \rightarrow \C $ is a holomorphic function satisfying $v(\tau)^2 = (c \tau + d )$. It is generated by 
\begin{align}
    S := & \left(\begin{pmatrix}
    0 & -1 \\ 1 & 0
\end{pmatrix}, \sqrt{\tau}\right) & T:= \left( \begin{pmatrix}
    1 & 1 \\ 0 & 1
\end{pmatrix},1 \right).
\end{align}
To an even lattice one associates a Weil representation of $\mathrm{Mp}(2,\Z)$. In this article, we will focus on a one-dimensional lattice $(\Z, 2p)$. The group ring $\C[\Z_{2p}]$ is a $\mathrm{Mp}(2,\Z)$-module via the Weil representation $\rho_{p}$, which is given by
\begin{align}
    \rho_{p}(T)_{a,b}:= &  \delta_{a,b} e^{2 \pi i \tfrac{a^{2}}{4p}} & \rho_{p}(S)_{a,b}:= \frac{e^{- \tfrac{i \pi}{4} \mathrm{Sign}(p)}}{\sqrt{|2p|}}e^{-i \pi \tfrac{a b}{p} },
\end{align}
where $a, b \in \Z_{2p}$. 

A real-analytic function $f: \mathbb{H} \rightarrow \C[\Z_{2p}] $ is called \textit{modular} of weight $k \in \frac{1}{2}\Z$ with respect to $\rho_{p}$ if 
\be 
(f\vert_{k, p}(\gamma, v))(\tau) = v(\tau)^{-2k} \rho_{p}^{-1}((\gamma,v))f(\gamma \cdot \tau)= f(\tau)
\ee 
for all $(\gamma,v) \in \mathrm{Mp}(2,\Z)$ and $\tau \in \mathbb{H}$. We denote by $M_{k,p}$, $\overline{M}_{k,p}$, $M_{k,p}^{!}$, and $\hat{M}_{k,p}$ the spaces of holomorphic, anti-holomorphic, weakly holomorphic modular forms, and weak harmonic Maass forms of weight $k$ with respect to $\rho_{p}$, respectively. Weakly holomorphic functions are allowed to have exponential growth at $\Q \cup \{\infty\}$. Weak harmonic Maass forms are real analytic modular forms of at most exponential growth at $\Q \cup \{\infty\}$ which are eigenfunctions of the Laplace operator. For details, see \cite{zagier2009ramanujan}. 

The antilinear differential operator $ \xi_{k} = 2i \mathrm{Im}(\tau)^{k} \frac{\overline{\partial}}{\partial \bar{\tau}} $ gives us a surjective map from $\hat{M}_{k,p}$ to $M_{2-k,-p}^{!}$. For a weight $2-k$ modular form $g$ we define $g^{*}$ as the following integral,
\be 
g^{*}(\tau) := -(2i)^{k-1} \int_{- \bar{\tau}}^{\infty} (z+ \tau)^{-k} \overline{g(-\bar{z})} \dd z.
\ee 
The function $g^{*}$ is such that $\xi_{k}g^{*} = g$. Therefore, we can extract the holomorphic part of a weak harmonic Maass form $\hat{f}$ by subtracting $ (\xi_{k} \hat{f})^{*} $ from it.

\begin{definition}
    A weakly holomorphic function $f^{+}:\mathbb{H} \rightarrow \C[\Z_{2p}]  $ is called a \textbf{vector-valued mock modular form} of weight $k \in \frac{1}{2}\Z$ with respect to $\rho_{p}$ and with shadow $g$ if it is the holomorphic part of a harmonic Maass form $\hat{f}$ of weight $k$ and with respect to $\rho_{p}$ which satisfies $g = \xi_{k}\hat{f}$. 
\end{definition}
\noindent We denote the space of mock modular forms of weight $k$ with respect to $\rho_{p}$ by $M_{k,p}^{+}$.
\begin{table}
    \centering
    \begin{tabular}{|c|c|}
    \hline
         $M_{k,p}^{!}$& Weakly holomorphic forms  \\
         $\hat{M}_{k,p}$& Weak harmonic Maass forms \\
         $M_{k,p}^{+}$ & Mock modular forms \\
         $M_{k,p}$ & Holomorphic forms\\
         $\overline{M}_{k,p}$ &  Anti-holomorphic forms\\
         \hline
    \end{tabular}
    \caption{Notation for spaces of modular functions.}
    \label{tab:modularfunctionnotation}
\end{table}
Table \ref{tab:modularfunctionnotation} summarizes the notation for the spaces of modular functions. The following commutative diagram, adapted from \cite{zagier2009ramanujan}, summarizes the relations among these spaces.
\be \label{commutativediagmocktheta}
\begin{tikzcd}[column sep=large]
    0 \arrow[r] & M_{k,p}^{!} \arrow[r,"\subset"] \arrow[d,"="] & \hat{M}_{k,p} \arrow[r, "\mathrm{Im}(\tau)^{k} \frac{\partial}{\partial \bar{\tau}}"  ] \arrow[d, "\cong"]& \overline{M}_{2-k,-p} \arrow[r] \arrow[d, "\text{c.c.}"', "\cong" ] & 0 \\
      0 \arrow[r] & M_{k,p}^{!} \arrow[r, "\subset"] & M_{k,p}^{+}  \arrow[r, "\text{Shadow}", ] & M_{2-k,-p}\arrow[r] &0
\end{tikzcd}
\ee 
Here c.c. denotes complex conjugation, and the isomorphism between $\hat{M}_{k,p}$ and $M_{k,p}^{+}$ is given by extracting the holomorphic part of a weak harmonic Maass form. 

For a positive integer $p$ and $a \in \Z_{2p}$, the weight $\frac{3}{2}$ unary theta functions $\theta^{1}_{p,a}\in M_{\frac{3}{2},p}$ are given by
\be 
\theta^{1}_{p,a}(\tau) = \sum_{n \in a + 2 p \Z} n q^{\frac{n^{2}}{4p}}.
\ee 
The false theta functions are the Eichler integrals of $\theta^{1}_{p,a}$, 
\be 
\psi_{p,a}(\tau) = \frac{1}{\sqrt{2 p \, i}}\int_{\tau}^{i \infty} \frac{\theta^{1}_{p,a}(z)}{\sqrt{ (z - \tau)}} \dd z.
\ee 
We can express the false theta functions as a $q$-series, with $q= e^{2 \pi i \tau}$ as follows, 
\be 
\psi_{p,a}(\tau) = \sum_{ n \in a + 2 p \Z }\mathrm{Sign}(n) q^{\frac{n^{2}}{4p}}.
\ee 
For more details on the modularity of false theta functions, see \cite{bringmann2019framework}. The $\hat{Z}$ invariant of negative-definite three-singular fibred Seifert manifolds can be written as a linear combination of false theta functions \cite{Cheng:2018vpl}. 

The mock conjecture from \cite{Cheng:2019uzc} states that for positive-definite three-singular fibred Seifert manifolds, the $\hat{Z}$ invariant is a mock theta function with a shadow given by a linear combination of weight $\frac{3}{2}$ unary theta functions whose Eichler integral is the $\hat{Z}$ invariant for the orientation-reversed manifold. This conjecture tells us how to make sense of the $q \rightarrow q^{-1}$ operation on false theta functions: 
\begin{tcolorbox}[colback=black!5!white, colframe=black!75!black,arc=0mm]
 For a positive integer $p$, and $a \in \Z_{2p}$, a \emph{dual false theta function} $\psi^{\vee}_{p,a}(\tau) \sim \psi_{p,a}(-\tau)$ is a vector-valued mock modular form of weight $\frac{1}{2}$ with respect to $\rho_{-p}$, with shadow $\tfrac{1}{\sqrt{p}}\theta^{1}_{p,a}$ and with Fourier coefficients in $\tfrac{1}{\mathcal{N}_{a}}\Z$ for some integer $\mathcal{N}_{a}$.   
\end{tcolorbox}

%%%%%%%%%%%%%%%%%%%%%%%%%%%%%%%%%%%%%%%

\subsection{The Resurgent Bridge: Borel--Mordell Integrals and Transseries} \label{resurgencebridge}
The $\hat{Z}$ invariants of three-manifolds provide a non-perturbative completion of $\mathrm{SL}(2,\C)$ Chern--Simons theory, and resurgent analysis makes this explicit. We first review the resurgent analysis for the negative Brieskorn sphere $\Sigma(p_{1},p_{2},p_{3})$, which has been studied extensively, and we refer the reader to \cite{Chung:2020efy, Gukov:2016njj, Andersen:2018khh} and the references therein for details. We then review how this structure is used to make sense of the $q \rightarrow q^{-1}$ operation for false theta functions \cite{Adams:2025aad, Costin:2023kla}.

The classical solutions of $\mathrm{SL}(2,\C)$ Chern--Simons theory on a three-manifold are given by the $\mathrm{SL}(2,\C)$ flat connections on that three-manifold. The perturbative expansion of the partition function around these classical solutions can be computed using standard techniques \cite{Axelrod:1991vq}, but the resulting series are typically divergent. Resurgence theory allows us to resum them and extract a non-perturbative partition function. For the negative Brieskorn sphere $\Sigma(p_{1},p_{2},p_{3})$, the Borel resummation of the perturbative expansion around the trivial flat connection can be written as the following Borel--Mordell integral, 
$$ \int e^{- 2 \pi i k \tfrac{u^{2}}{4p}}   \frac{\sin(  \tfrac{\pi u}{p_{1}} )\sin( \tfrac{\pi u}{p_{2}}) \sin( \tfrac{\pi u}{p_{3}})}{ \sin(\pi u)}\mathrm{d} u,$$ 
where $k$ is the Chern--Simons level, and $p = p_{1}p_{2} p_{3}$. 
The non-perturbative contributions from other (non-abelian) flat connections are captured by the poles of the integrand. We can write the transseries expansion of $\hat{Z}(\Sigma(p_1, p_2, p_3);q)$ as
\be \label{zhattransseries}
\hat{Z}(\Sigma(p_1, p_2, p_3);q) = \sqrt{\tfrac{i k}{2p}}\int_{e^{i \phi} \R}  e^{- 2 \pi i k \tfrac{u^{2}}{4p}}   \frac{\sin(  \tfrac{\pi u}{p_{1}} )\sin( \tfrac{\pi u}{p_{2}}) \sin( \tfrac{\pi u}{p_{3}})}{ \sin(\pi u)}\mathrm{d} u + \sum_{\alpha \in \mathcal{M}_{nab}} c_{\alpha } \mathcal{Z}_{\alpha}(k) ,
\ee 
where $\phi = -\frac{1}{2}\mathrm{Arg}(i k)$, $q= e^{\frac{2 \pi i}{k}}$, and $\alpha$ is summed over non-trivial flat connections, which in the case of the Brieskorn sphere are all non-abelian. The contribution $\mathcal{Z}_{\alpha}(k)$ from the non-abelian flat connections is given by 
\be 
\mathcal{Z}_{\alpha}(k)  =\sum_{n \in \Z } m_{(\alpha,n),0} \tilde{q}^{\,n + \mathrm{CS}(\alpha)} Z^{\text{pert}}_{\alpha}\left(\tfrac{1}{k} \right),
\ee 
where $\tilde{q} = e^{-2 \pi i k}$, $ m_{(\alpha,n),0}$ are the Stokes coefficients and $n + \mathrm{CS}(\alpha)$ is the integer lift of the Chern--Simons action of the flat connection $\alpha$. In the case of the three-singular fibred Brieskorn sphere, the perturbative series $Z^{\text{pert}}_{\alpha}\left(\tfrac{1}{k} \right)$ has just one non-zero term at order $\sqrt{k}$ and all higher-order terms vanish. This decomposition \eqref{zhattransseries} reveals the resurgent structure of $\hat{Z}(\Sigma(p_1, p_2, p_3);q)$, which schematically takes the form
\be 
q^{\delta} Z(q) = \sqrt{\tfrac{i k}{2p}} \text{ ``Borel--Mordell integral''} + \sqrt{k} \sum_{\alpha \in \mathcal{M}_{nab}}  c_{\alpha } \tilde{q}^{\tilde{\Delta}_{\alpha}}  W_{\alpha}(\tilde{q}),
\ee 
where $Z(q) \in \Z[[q]]$, $W_{\alpha}(\tilde{q}) \in \Z[[\tilde{q}]]$, and $\tilde{\Delta}_{\alpha} = \mathrm{CS}(\alpha) \mod \Z$. The same resurgent structure is already present in the false theta functions, the building blocks of $\hat{Z}(\Sigma(p_1, p_2, p_3);q)$, where it takes the form of the identity 
\be \label{resurgenceeqfalsetheta}
\psi_{p,a}(\tfrac{1}{k}) + \sqrt{\frac{ k }{i}} \sum_{b =1}^{p-1} \sqrt{\frac{2}{p}} \sin\left(\frac{\pi a b}{p}\right) \psi_{p,b}(-k)= \sqrt{\tfrac{i k}{2p}} \int_{e^{i \phi} \R}  e^{- 2 \pi i k \tfrac{u^{2}}{4p}} \frac{\sin(\pi u\frac{p-a}{p})}{ \sin(\pi u )} \dd u.
\ee 
Note that $\psi_{p,a}(\tfrac{1}{k}) \in q^{\Delta_{a}}\Z[[q]]$ with $\Delta_{a} = \tfrac{a^{2}}{4 p}\mod \Z $. The identity \eqref{resurgenceeqfalsetheta} reveals further structure: the $q$-series and the $\tilde{q}$-series are both false theta functions. This additional structure follows from the quantum modularity of false theta functions, and it is present in $\hat{Z}$ as well \cite{Cheng:2018vpl}. The relation between resurgence and quantum modularity in the context of false theta functions was studied systematically in \cite{Cheng:2026tnu}. More generally, this interplay has been developed into a general framework of modular resurgence \cite{Fantini:2024ihf, Fantini:2024snx, mcspirit2025quantum, Fantini:2025wap, Fantini:2026ale}.

In \cite{Costin:2023kla}, Costin, Dunne, Gruen, and Gukov used the resurgent structure of the false theta functions discussed above to make sense of the $q \rightarrow q^{-1}$ operation on false theta functions. The Borel--Mordell integral on the right-hand side of \eqref{resurgenceeqfalsetheta} can be analytically continued from $k$ in the lower half-plane to the upper half-plane. Using this analytic continuation we can perform the $q \rightarrow q^{-1}$ operation on the Borel--Mordell integral
\begin{align}
 \int_{e^{i \phi} \R}  e^{- 2 \pi i k \tfrac{u^{2}}{4p}} \frac{\sin(\pi u\frac{p-a}{p})}{ \sin(\pi u )} \dd u \hspace{0.5cm} \xleftrightarrow{\hspace{1.75cm}} \hspace{0.5cm}   \int_{e^{i \phi} \R}  e^{- 2 \pi i k \tfrac{u^{2}}{4p}} \frac{\sinh(\pi u\frac{p-a}{p})}{ \sinh(\pi u )} \dd u. 
\end{align}
The resurgent structure of the false theta functions should be preserved under $q \rightarrow q^{-1}$. This imposes the following conditions on the dual false theta functions $\psi^{\vee}_{p, a}$: 
\begin{tcolorbox}[colback=black!5!white, colframe=black!75!black,arc=0mm]
\begin{enumerate}[leftmargin=0.2cm, label=\arabic*]
    \item The $q$-series expansion of $\psi^{\vee}_{p,a}$ satisfies $\psi^{\vee}_{p,a}(\tfrac{1}{k}) \in \tfrac{q^{\Delta^{\vee}_{a}}}{\mathcal{N}_{a}} \Z[[q]] $ with $\Delta^{\vee}_{a} = - \frac{a^{2}}{4 p} \mod \Z$ for some $\mathcal{N}_{a} \in \Z$.
    \item The functions $\psi^{\vee}_{p,a}$ satisfy
    \be 
\psi^{\vee}_{p,a}(\tfrac{1}{k}) + \sqrt{i k} \sum_{b =1}^{p-1} \sqrt{\frac{2}{p}} \sin\left(\frac{\pi a b}{p}\right) \psi_{p,b}^{\vee}(-k) =  \sqrt{\tfrac{i k}{2p}} \int_{e^{i \phi} \R}  e^{- 2 \pi i k \tfrac{u^{2}}{4p}} \frac{\sinh(\pi u\frac{p-a}{p})}{ \sinh(\pi u )} \dd u.
\ee 
\end{enumerate}
\end{tcolorbox}
In \cite{Adams:2025aad, Costin:2023kla, Adams:2025qgj}, using numerical techniques, the first $30$ to $60$ coefficients of some dual false theta functions and analogous $q$-series were obtained for $\hat{Z}$ for orientation-reversed Brieskorn spheres. 

%%%%%%%%%%%%%%%%%%%%%%%%%%%%%%%%%%%%%%%

\subsection{Two Bridges are the Same}\label{equivalenceofbridgessection}
Although the two bridges to the other side ($q \rightarrow q^{-1}$) have different origins, they turn out to be the same bridge. In this subsection, we show that the modular and resurgent characterizations of the dual false theta functions discussed above are equivalent. 

\begin{theorem}\label{equivalenceofbridges}
    For a positive integer $p$ and $a \in \Z_{2p}$, a vector-valued function $\psi_{p, a}^{\vee}$ satisfies the resurgent characterization of dual false theta functions if and only if it satisfies the modular characterization of dual false theta functions.
\end{theorem}
\begin{proof} \textbf{Modular $\implies$ resurgent.} Suppose $\psi_{p,a}^{\vee}$ is a vector-valued mock modular form of weight $\frac{1}{2}$ with respect to $\rho_{-p}$, with shadow $\tfrac{1}{\sqrt{p}}\theta^{1}_{p,a}$ and with Fourier coefficients in $\tfrac{1}{\mathcal{N}_{a}}\Z$ for some integer $\mathcal{N}_{a}$.

Let us denote by $\hat{\psi}_{p,a}^{\vee}$ the harmonic Maass form (modular completion) associated with $\psi_{p,a}^{\vee}$. It is given by 
\be 
\hat{\psi}_{p,a}^{\vee}(\tau, \bar{\tau}) = \psi_{p,a}^{\vee}(\tau) - \sqrt{\frac{-i}{2p}} \int_{- \bar{\tau}}^{\infty}  \frac{\overline{\theta_{p,a}^{1}(-\bar{z})}}{ \sqrt{z+ \tau} } \dd z.
\ee 
\textit{$T$-transformation.} Under $\tau \mapsto \tau +1$, the completion $\hat{\psi}_{p,a}^{\vee}$ transforms as $\hat{\psi}_{p,a}^{\vee}(\tau +1) = e^{- 2 \pi i \tfrac{a^{2}}{4p}} \hat{\psi}_{p,a}^{\vee}(\tau)$. The non-holomorphic part transforms in the same way:
\be 
 \int_{- \bar{\tau}-1}^{\infty}  \frac{\overline{\theta_{p,a}^{1}(-\bar{z})}}{ \sqrt{z+ \tau+1} } \dd z  =   \int_{- \bar{\tau}}^{\infty}  \frac{\overline{\theta_{p,a}^{1}(-\bar{z}+1)}}{ \sqrt{z+ \tau} } \dd z =  e^{- 2 \pi i \tfrac{a^{2}}{4p}} \int_{- \bar{\tau}}^{\infty}  \frac{\overline{\theta_{p,a}^{1}(-\bar{z})}}{ \sqrt{z+ \tau} } \dd z.
\ee 
Hence the holomorphic part satisfies $\psi_{p,a}^{\vee}(\tau+1) = e^{- 2 \pi i \tfrac{a^{2}}{4p}} \psi_{p,a}^{\vee}(\tau)$. Together with the integrality assumption, this means that $\psi^{\vee}_{p, a}(\tfrac{1}{k}) \in \tfrac{q^{\Delta^{\vee}_{a}}}{\mathcal{N}_{a}} \Z[[q]] $ with $\Delta^{\vee}_{a} = - \frac{a^{2}}{4 p} \mod \Z$, which is the first resurgent condition.\footnote{We would like to thank Miranda Cheng for pointing this out.}

\textit{$S$-transformation.} Under $\tau \mapsto -1/\tau$, $\hat{\psi}_{p,a}^{\vee}$ transforms as
\be 
\hat{\psi}_{p,a}^{\vee}(\tau, \bar{\tau}) = - \sqrt{\frac{i}{\tau}} \sum_{b =1}^{p-1} \sqrt{\frac{2}{p}} \sin\left(\frac{\pi a b}{p}\right)  \hat{\psi}_{p,b}^{\vee}(-1/\tau,-1/\bar{\tau}). 
\ee 
Collecting the holomorphic parts on the left-hand side and the non-holomorphic part on the right-hand side, we get 
\begin{align} \label{stransformdftholnonholsep}
& \psi_{p,a}^{\vee}(\tau) + \sqrt{\frac{i}{\tau}} \sum_{b =1}^{p-1} \sqrt{\frac{2}{p}} \sin\left(\frac{\pi a b}{p}\right)  \psi_{p,b}^{\vee}(-1/\tau) \endline  = &  \sqrt{\frac{- i}{2p}} \left[ \int_{- \bar{\tau}}^{\infty}  \frac{\overline{\theta_{p,a}^{1}(-\bar{z})}}{ \sqrt{z+ \tau} } \dd z + \sqrt{\frac{i}{\tau}} \sum_{b =1}^{p-1} \sqrt{\frac{2}{p}} \sin\left(\frac{\pi a b}{p}\right)   \int_{1/\bar{\tau}}^{\infty}  \frac{\overline{\theta_{p,b}^{1}(-\bar{z})}}{ \sqrt{z-1/ \tau} } \dd z  \right].
\end{align}
We now show that the right-hand side is a Borel--Mordell integral. Changing variables $z \mapsto - 1/z$ in the last integral gives
\be 
\sqrt{\frac{i}{\tau}}\int_{1/\bar{\tau}}^{\infty}  \frac{\overline{\theta_{p,b}^{1}(-\bar{z})}}{ \sqrt{z-1/ \tau} } \dd z = \sqrt{\frac{i}{\tau}} \int_{- \bar{\tau}}^{0}  \frac{\overline{\theta_{p,b}^{1}(1/\bar{z})}}{ \sqrt{-1/z -1/ \tau} } \frac{1}{z^{2}} \dd z =  \int_{0}^{- \bar{\tau}} \frac{ \overline{\theta_{p,b}^{1}(1/\bar{z})}   }{\sqrt{z + \tau}} \left(\frac{i}{z}\right)^{\frac{3}{2}} \dd z,
\ee 
so the bracket in \eqref{stransformdftholnonholsep} becomes, 
\begin{align}
    & \int_{- \bar{\tau}}^{\infty}  \frac{\overline{\theta_{p,a}^{1}(-\bar{z})}}{ \sqrt{z+ \tau} } \dd z +  \int_{0}^{- \bar{\tau}} \frac{ 1}{\sqrt{z + \tau}}  \left(\frac{i}{z}\right)^{\frac{3}{2}} \sum_{b =1}^{p-1} \sqrt{\frac{2}{p}} \sin\left(\frac{\pi a b}{p}\right) \overline{\theta_{p,b}^{1}(1/\bar{z})}   \dd z.
\end{align}
Since $\theta_{p,a}^{1}$ is modular of weight $3/2$ with respect to $\rho_{p}$, the sum over $b$ in the second integrand equals $\overline{\theta_{p,a}^{1}(-\overline{z})}$. The two integrals then combine into a single integral over $\R_{>0}$:
\begin{align}
     \int_{- \bar{\tau}}^{\infty}  \frac{\overline{\theta_{p,a}^{1}(-\bar{z})}}{ \sqrt{z+ \tau} } \dd z +  \int_{0}^{- \bar{\tau}} \frac{ \overline{\theta_{p,a}^{1}(-\bar{z})}}{\sqrt{z + \tau}}  \dd z =  &\int_{0}^{\infty}  \frac{\overline{\theta_{p,a}^{1}(-\bar{z})}}{ \sqrt{z+ \tau} } \dd z = \sum_{n \in a + 2 p \Z}  \int_{0}^{\infty}  \frac{ n e^{ i \pi z \frac{n^{2}}{2p}    }}{ \sqrt{z+ \tau} } \dd z,
\end{align}
where in the last step we used the definition of $\theta_{p, a}^{1}$ and exchanged integration and summation.

Substituting $z= \frac{2 i  u}{ n} - \frac{u^{2}}{n^{2}\tau} $ and summing the resulting geometric series, we find
\begin{align}
    \sum_{n \in a + 2 p \Z}  \int_{0}^{\infty}  \frac{ n e^{ i \pi z \frac{n^{2}}{2p}    }}{ \sqrt{z+ \tau} } \dd z  = & \sqrt{-1/\tau}  \sum_{n \in a + 2 p \Z} 2 \int_{0}^{\infty} \mathrm{Sign}(n) e^{- \frac{i \pi u^{2}}{2 p \tau}} e^{- \frac{\pi u |n|}{p}} \dd u \endline
    = &2 \sqrt{-1/\tau}    \int_{0}^{\infty}  e^{- \frac{i \pi u^{2}}{2 p \tau}} \sum_{n \in a + 2 p \Z} \mathrm{Sign}(n)e^{- \frac{\pi u |n|}{p}} \dd u \endline
    = &2 \sqrt{-1/\tau}    \int_{0}^{\infty}  e^{- \frac{i \pi u^{2}}{2 p \tau}} \frac{\sinh(\pi u \tfrac{p-a}{p})}{\sinh(\pi u)} \dd u .
\end{align}
Therefore,
\begin{align} \label{psidualstransform1}
\psi_{p,a}^{\vee}(\tau) + \sqrt{\frac{i}{\tau}} \sum_{b =1}^{p-1} \sqrt{\frac{2}{p}} \sin\left(\frac{\pi a b}{p}\right)  \psi_{p,b}^{\vee}(-1/\tau)  = &  2\sqrt{\frac{i}{2 p \tau}} \int_{0}^{\infty}  e^{- \frac{i \pi u^{2}}{2 p \tau}} \frac{\sinh(\pi u \tfrac{p-a}{p})}{\sinh(\pi u)} \dd u 
\end{align}
Setting $\tau = \tfrac{1}{k}$, using that the integrand is even in $u$, and rotating the contour, this becomes
\begin{align}
    \psi_{p,a}^{\vee}(\frac{1}{k}) + \sqrt{ik} \sum_{b =1}^{p-1} \sqrt{\frac{2}{p}} \sin\left(\frac{\pi a b}{p}\right)  \psi_{p,b}^{\vee}(-k)  = &  \sqrt{\tfrac{i k}{2p}} \int_{e^{i \phi} \R}  e^{- 2 \pi i k \tfrac{u^{2}}{4p}} \frac{\sinh(\pi u\frac{p-a}{p})}{ \sinh(\pi u )} \dd u,
\end{align}
which is the second resurgent condition.

\textbf{Resurgent $\implies$ modular.} Now suppose $\psi_{p, a}^{\vee}$ satisfies the resurgent characterization. Let $\chi_{p,a}$ be any vector-valued mock modular form of weight $\frac{1}{2}$ with respect to $\rho_{-p}$, with shadow $\tfrac{1}{\sqrt{p}}\theta^{1}_{p,a}$ and with Fourier coefficients in $\tfrac{1}{\mathcal{N}_{a}}\Z$ for some integer $\mathcal{N}_{a}$; as we will see in Section \ref{dualfalseformulas}, such a function exists for all $p$, $a$. Set $f_{p,a} := \psi_{p,a}^{\vee} - \chi_{p,a}$.

By the first resurgent condition, $\psi_{p,a}^{\vee}(\tau+1) = e^{- 2 \pi i \tfrac{a^{2}}{4p}} \psi_{p,a}^{\vee}(\tau)$, and $\chi_{p,a}$ transforms in the same way, so 
$$f_{p,a}(\tau+1) = e^{- 2 \pi i \tfrac{a^{2}}{4p}}f_{p,a}(\tau).$$

By the first part of the proof, $\chi_{p, a}$ satisfies \eqref{psidualstransform1}, and by the second resurgent characterization so does $\psi_{p, a}^{\vee}$. Subtracting the two, the Borel--Mordell integrals cancel, and under $\tau \mapsto - 1/\tau$ we obtain
\begin{align} 
f_{p,a}(\tau) = - \sqrt{\frac{i}{\tau}} \sum_{b =1}^{p-1} \sqrt{\frac{2}{p}} \sin\left(\frac{\pi a b}{p}\right)  f_{p,b}(-1/\tau).
\end{align}
Since $S$ and $T$ generate $\mathrm{Mp}(2,\Z)$, $f_{p, a}$ is a weakly holomorphic modular form of weight $\frac{1}{2}$ with respect to $\rho_{-p}$. Adding a weakly holomorphic modular form to a mock modular form gives another mock modular form with the same shadow. Hence, $\psi_{p, a}^{\vee} = \chi_{p, a}+ f_{p, a}$ is a vector-valued mock modular form of weight $\frac{1}{2}$ with respect to $\rho_{-p}$, with shadow $\tfrac{1}{\sqrt{p}}\theta^{1}_{p, a}$, and its coefficients are in $\frac{1}{\mathcal{N}_{a}}\Z$ by the first resurgent condition.
\end{proof}

%%%%%%%%%%%%%%%%%%%%%%%%%%%%%%%%%%%%%%%

\section{Dual False Theta Functions}\label{dualfalseformulas}
In this section, we will write down one family of dual false theta functions for any positive integer $p$, and one family for special values of $p$. The family that gives dual false theta functions for any $p$ is from \cite{li2022mock}, while the family for special values of $p$ is, to our knowledge, new.

%%%%%%%%%%%%%%%%%%%%%%%%%%%%%%%%%%%%%%%

\subsection{Li--Schwagenscheidt Family of Mock Theta Functions}

In \cite{li2022mock}, Li and Schwagenscheidt constructed a vector-valued mock modular form of weight half for $\rho_{-p}$ and with shadow $\tfrac{1}{\sqrt{p}}\theta^{1}_{p, a}$ and rational coefficients of bounded denominator. That is, they gave a dual false theta function. 

Suppose $2p$ is not a square, and let $ u + v \sqrt{2p} $ be the smallest totally positive unit $>1$ of $\Q(\sqrt{2p})$ such that $v$ is even and $\mathrm{lcm}(2p,4) \mid (u-1)$. By a totally positive unit, we mean $u\pm v \sqrt{2p}>0 $. Then the dual false theta functions are given by: 
\be 
\psi^{\vee, \eta^{3}}_{p,a}(q) = \frac{1}{q^{\tfrac{1}{8}}(q)_{\infty}^{3}} \sum_{\substack{x \in a+ 2 p \Z  \\y \in \Z \\ 0 < - \tfrac{x}{y} < \tfrac{v p}{u} }} \mathrm{Sign}(x) \left( \frac{-4}{y}\right) \left[ \frac{y}{2} + \frac{v x}{u-1} \right] q^{\tfrac{y^{2}}{8}  - \tfrac{x^{2}}{4p}  },
\ee 
where $\left( \tfrac{-4}{y}\right)$ is the Kronecker symbol. Note that $\psi^{\vee, \eta^{3}}_{p,a}(q)$ has coefficients in $\frac{1}{u-1}\Z$.

Suppose $2p$ is a square, and let $E_{2}(\tau)$ be the Eisenstein series given by 
\be 
E_{2}(\tau) = 1-24 \sum_{n =1}^{\infty} \sigma_{1}(n) q^{n} =1-24 \sum_{n =1}^{\infty} \frac{n q^{n}}{1-q^{n}},
\ee 
where $\sigma_{1}$ is the divisor function. Then for $a = b\sqrt{p/2}$ the dual false theta functions are given by,
\begin{align}
    & \psi^{\vee, \eta^{3}}_{p,a}(q) \endline = & \frac{1}{\eta(\tau)^{3}} \Bigg[ \left(\frac{-4}{b} \right) \left( \frac{E_{2}(\tau)}{12\sqrt{2p}} - \sqrt{2p} \, \mathbb{B}_{2}\left(\frac{b}{2 \sqrt{2p}} \right)  \right)+  \sum_{\substack{x \in a+ 2p \Z , \, y \in \Z  \\ y > \tfrac{\sqrt{2}|x|}{\sqrt{p}}}} \left(\frac{-4}{y} \right) \mathrm{Sign}(x) y q^{ \tfrac{y^{2}}{8} - \tfrac{x^{2} }{4p}}   \Bigg],
\end{align}
and for $a \notin \sqrt{p/2}\Z$,  
\begin{align}
    \psi^{\vee, \eta^{3}}_{p,a}(q) = & \frac{1}{q^{\tfrac{1}{8}}(q)_{\infty}^{3}}   \sum_{\substack{x \in a+ 2p \Z, \, y \in \Z \\ y > \tfrac{\sqrt{2}|x|}{\sqrt{p}}}} \left(\frac{-4}{y} \right) \mathrm{Sign}(x) y \, q^{ \tfrac{y^{2}}{8} - \tfrac{x^{2} }{4p}} ,
\end{align}
where $(\frac{-4}{y})$ and $(\frac{-4}{b})$ are the Kronecker symbols, and $ \mathbb{B}_{2}$ is a one-periodic function which agrees with the Bernoulli polynomial on the interval $[0,1)$.

%%%%%%%%%%%%%%%%%%%%%%%%%%%%%%%%%%%%%%%

\subsection*{Li--Schwagenscheidt family from regularization of \texorpdfstring{$\,\eta(\tau)^{3}\, \psi_{p,a}(- \tau)$}{eta(tau)cubed psi(-tau)}}

Since $\eta(\tau)^{3} = \tfrac{1}{2} \sum_{y \in 1 + 2 \Z} (\tfrac{-4}{y})y \, q^{\tfrac{y^{2}}{8}}$, the formal product $\eta(\tau)^{3} \psi_{p,a}(- \tau)$ reads
\be \label{divergentsumetacubed}
\eta(\tau)^{3} \psi_{p,a}(- \tau) \sim \frac{1}{2}\sum_{\substack{x \in a+ 2 p \Z \\ y \in 1 + 2 \Z }} \mathrm{Sign}(x) \left(\frac{-4}{y}\right) y \, q^{\frac{y^{2}}{8} - \frac{x^{2}}{4p} } .
\ee 
This sum is divergent, since the exponent $\frac{y^{2}}{8} - \frac{x^{2}}{4p} $ is indefinite. We now regularize it and show that the result is the Li--Schwagenscheidt family of mock theta functions. 

Consider the even lattice $L = \left(\Z^{2}, \mathrm{diag}(-2p, 4)\right)$ of signature $(1,1)$. Let $ A_{L}:= L^{*}/L$ be the discriminant form, where $L^{*}$ is the dual lattice. Denote by $O(\R^{1,1})$ the orthogonal group of $L \otimes \R$ and by $O^{+}(\R^{1,1})$ its identity component. Let $O(L) \subset \mathrm{GL}(2,\Z)$ be the orthogonal group of $L$, let $\Gamma_{L} \subset O(L)$ be the subgroup acting trivially on $A_{L}$, and set $\Gamma^{+}_{L} := \Gamma_{L}\cap O^{+}(\R^{1,1})$. Explicitly,
\be 
\Gamma^{+}_{L} = \left\{ \begin{pmatrix}
    u & 2v \\ p v & u
\end{pmatrix} \in \mathrm{SL}(2,\Z) \bigg\vert v \text{ is even }, u \geq 1 ,\mathrm{lcm}(2p,4) \mid (u-1) \right\}.
\ee 
If $\sqrt{2p} \in \Z$, then the group $\Gamma^{+}_{L}$ is trivial. If $\sqrt{2p} \notin \Z$, it is infinite cyclic, generated by 
$$\gamma= \begin{pmatrix}
    u & 2v \\ p v & u
\end{pmatrix},$$
where $ u + v \sqrt{2p} $ is the smallest totally positive unit $>1$ of $\Q(\sqrt{2p})$ with $v$ even and $\mathrm{lcm}(2p,4) \mid (u-1)$. 

\textbf{The case} $\sqrt{2p} \notin \Z$. An element $\begin{pmatrix}
    u_1 & 2 v_1 \\ p v_1 & u_1
\end{pmatrix} \in \Gamma^{+}_{L}$ acts on $(x,y)\in (a + 2 p \Z)\times (1 + 2 \Z)$, by 
$$ (x,y) \mapsto (x,y) \begin{pmatrix}
    u_1 & 2 v_1 \\ p v_1 & u_1
\end{pmatrix} = (u_{1} x+ p v_{1} y , 2 v_1 x + u_1 y), $$
which preserves the exponent $ \frac{y^{2}}{8} - \frac{x^{2}}{4p}  $. Writing $(x_{n},y_{n}):= (x,y)\cdot \gamma^{n}$, we can therefore organize the sum \eqref{divergentsumetacubed} into $\Gamma^{+}_{L}$-orbits:
\begin{align}
    \eta(\tau)^{3} \psi_{p,a}(- \tau) \sim \frac{1}{2}\sum_{\substack{x \in a+ 2 p \Z \\ y \in 1 + 2 \Z } \big/ \Gamma_{L}^{+}  } \sum_{n \in \Z} \mathrm{Sign}(x_{n}) \left(\frac{-4}{y_{n}}\right) y_{n} \, q^{\frac{y^{2}}{8} - \frac{x^{2}}{4p} } .
\end{align}
As a set of orbit representatives, we may take the union of two fundamental domains, 
\begin{align*}
    \mathcal{F}_{-} = & \left\{ (x,y)\in (a+2 p \Z)\times (1+2 \Z) \Big\vert \, \frac{u}{2v}< \frac{x}{y}\right\}, \\
    \mathcal{F}_{+} = & \left\{ (x,y)\in (a+2 p \Z)\times (1+2 \Z) \Big \vert \, 0< - \frac{x}{y}< \frac{v p}{u}\right\},
\end{align*}
on which the exponent $ \frac{y^{2}}{8} - \frac{x^{2}}{4p} $ is negative and positive respectively.

Along each orbit, the divisibility conditions on $u$ and $v$ imply $y_{n} = y \mod 4$, and hence $(\tfrac{-4}{y_{n}}) = (\tfrac{-4}{y})$. For the sign of $x_{n}$, we have $\mathrm{Sign}(x_{n}) = \mathrm{Sign}(x)$ for all $n$ on $\mathcal{F}_{-}$, while on $\mathcal{F}_{+}$ we have $\mathrm{Sign}(x_{n}) = \mathrm{Sign}(x)$ for $n \leq 0$ and $\mathrm{Sign}(x_{n}) = -\mathrm{Sign}(x)$ for $n>0$. Therefore,
\begin{align}
  \eta(\tau)^{3} \psi_{p,a}(- \tau) \sim & \frac{1}{2} \sum_{\substack{ x \in a + 2 p \Z, y \in 1 + 2 \Z \\ \frac{u}{2 v} < \frac{x}{y}}}   \mathrm{Sign}(x) \left(\frac{-4}{y}\right) \left(\sum_{n \in \Z} y_{n} \right)q^{\frac{y^{2}}{8} - \frac{x^{2}}{4p} } \endline 
  & + \frac{1}{2}\sum_{\substack{ x \in a + 2 p \Z, y \in 1 + 2 \Z \\ 0< -\frac{x}{y}< \frac{v p}{u} }}   \mathrm{Sign}(x) \left(\frac{-4}{y}\right) \left( \sum_{n=-\infty}^{0} y_{n} - \sum_{n = 1}^{\infty} y_{n}\right)q^{\frac{y^{2}}{8} - \frac{x^{2}}{4p} } .
\end{align}
The inner sums are divergent geometric series: with $\lambda = u + v \sqrt{2p}$, each $y_{n}$ is a linear combination of $\lambda^{n}$ and $\lambda^{-n}$. We regularize them using $\sum_{n=1}^{\infty} z^{n} = \frac{z}{1-z} = - \sum_{n = -\infty}^{0}z^{n}$, which gives 
\begin{align}
    \sum_{n \in \Z} y_{n}  =& 0  &    \sum_{n=-\infty}^{0} y_{n}- \sum_{n = 1}^{\infty} y_{n} = & 2\left[ \frac{y}{2} + \frac{v x}{u-1} \right].
\end{align}
Hence only $\mathcal{F}_{+}$ contributes, and 
$$ \eta(\tau)^{3} \psi_{p,a}(- \tau) \xrightarrow[]{\mathrm{reg.}}  \sum_{\substack{x \in a+ 2 p \Z  \\y \in \Z \\ 0 < - \tfrac{x}{y} < \tfrac{v p}{u} }} \mathrm{Sign}(x) \left( \frac{-4}{y}\right) \left[ \frac{y}{2} + \frac{v x}{u-1} \right] q^{\tfrac{y^{2}}{8}  - \tfrac{x^{2}}{4p}  }.$$
Dividing by $\eta(\tau)^{3}$, we recover exactly the Li--Schwagenscheidt family $\psi_{p,a}^{\vee, \eta^{3}}$.

\textbf{The case} $\sqrt{2p} \in \Z$. Here $\Gamma^{+}_{L}$ is trivial, and a different regularization reproduces the Li--Schwagenscheidt family; we only sketch it. The sums over the two time-like cones $ \pm y > \frac{\sqrt{2}|x|}{\sqrt{p}}$ are convergent and contribute equally, while the sum over the space-like region $ |y| < \frac{\sqrt{2}|x|}{\sqrt{p}}$ vanishes coefficient by coefficient. The remaining divergent sums lie on the light cone, and can be regularized by slightly deforming the quadratic form $ \tfrac{y^{2}}{8} - \tfrac{x^{2}}{4p} $. Applying zeta regularization along one light-cone direction gives the $E_{2}$ term, and Abel regularization along the other gives the $\mathbb{B}_{2}$ term.

It would be interesting to generalize this $\eta^{3}$-regularization to $\hat{Z}$ for positive-definite plumbed manifolds. 

%%%%%%%%%%%%%%%%%%%%%%%%%%%%%%%%%%%%%%%

\subsection{Special Arithmetic Family}
In \cite{Park:2021ufu}, Park conjectured a regularized surgery formula for $\hat{Z}$ in cases where the Gukov--Manolescu formula \cite{Gukov:2019mnk} fails to produce a $q$-series. It was shown in \cite{Cheng:2024vou} that the $\hat{Z}$ invariants of positive Brieskorn spheres obtained from this conjectured formula can be written as a linear combination of the ratios of Zwegers' indefinite theta function \cite{zwegers2008mock} and the theta function, an ansatz first proposed in \cite{Cheng:2019uzc}. However, \cite{Harichurn:2025suf} showed that the $\hat{Z}$ obtained from the conjectured regularized surgery formula is incompatible with the resurgent characterization of dual false theta functions. Motivated by these works, we construct dual false theta functions as ratios of Zwegers' indefinite theta function and the Dedekind eta function.

Let $A=\mathrm{diag}(-2p, 3)$ and $\beta = (0,\tfrac{1}{6})^{T}$, and define the bilinear form $B(m,n) = m^{T}A n$, for $m, n \in \R^{2}$. For $\mu \in \beta + A^{-1}\Z^{2}$, and $c_{1}$, $c_{2} \in \R^{2}$ satisfying $B(c_1, c_{2}) , B(c_{1}, c_{1}), B(c_{2}, c_{2}) <0$, Zwegers' indefinite theta function is
\be 
\theta_{A,\mu, \beta, c_{1}, c_{2} }(\tau) = \sum_{n \in \mu + \Z^{2}} \left( \mathrm{Sign}\left( B(c_{1}, n) \right) - \mathrm{Sign}\left( B(c_{2}, n) \right) \right) e^{2 \pi i B(\beta, n)} q^{\tfrac{B(n,n)}{2}}.
\ee 

\begin{tcolorbox}[colback=black!5!white, colframe=black!75!black,arc=0mm]
\emph{Let $p$ be a positive integer such that $\sqrt{6p} \notin \Z$, and suppose the fundamental solution $(u, v)$ of the Pell equation $u^{2} - 6 p v^{2}=1$ satisfies: $v$ is even and $ \left(\frac{12}{u}\right) u =- 1 \mod 2p$.} For $a \in \Z_{2p}$, set $\mu_{0}(a) = (\tfrac{a}{2p},\tfrac{1}{6})$, $\mu_{2}(a) = (\tfrac{a}{2p},\tfrac{5}{6})$, $c_{1}= (1,0)^{T}$, and $c_{2}=(u,2 p v)^{T}$. We define the dual false theta functions $\psi^{\vee, \eta}_{p,a}$ by, 
\be \label{sparthdualfalsedef}
\psi^{\vee, \eta}_{p,a}(\tau) := \frac{1}{2\eta(\tau)} \left( e^{- \tfrac{i \pi }{6}}\theta_{A,\mu_{0}(a), \beta, c_{1}, c_{2} }(\tau) + e^{ \tfrac{i \pi }{6}} \theta_{A,\mu_{2}(a), \beta, c_{1}, c_{2} }(\tau) \right).
\ee 
\end{tcolorbox}
The first $20$ values of $p>1$ for which the arithmetic condition is satisfied are:
\begin{align*}
    \{7,13,25,31,37,49,55,61,73,79,97,103,109,121,
   127,151,157,169,175, 181\}.
\end{align*}

\begin{theorem}\label{arthfamthm}
$ \psi^{\vee, \eta}_{p,a}$ is a vector-valued mock modular form of weight $\frac{1}{2}$ with respect to $\rho_{-p}$, with shadow $-\frac{2}{\sqrt{p}}\theta^{1}_{p,a}$.
\end{theorem}
\begin{proof}
We first compute the shadow of $\psi^{\vee, \eta}_{p,a}$ and then show that the modular completion of $\psi^{\vee, \eta}_{p,a}$ is modular of weight $\frac{1}{2}$ with respect to $\rho_{-p}$.

Note that $B(c_{i}, c_{i}) = -2p$ for $i \in \{1,2\}$. The sublattice $\{ \lambda \in \Z^{2} \vert B(c_1, \lambda) =0 \}$ is generated by $c_{1, \perp}=(0,1)$, and $\{ \lambda \in \Z^{2} \vert B(c_2, \lambda) =0 \}$ is generated by $c_{2, \perp} = (3v,u)$. Since $u^{2}- 6 p v^{2}=1$, the quotient $\Z^{2}/ (c_{i}\Z \oplus c_{i, \perp}\Z)$ is trivial for both $i =1 $ and $i =2$, which makes the shadow computation particularly simple.

\noindent We denote by $\hat{\psi}^{\vee, \eta}_{p,a}$ the modular completion of $\psi^{\vee, \eta}_{p,a}$. It is given by 
\begin{align}
   \hat{\psi}^{\vee, \eta}_{p,a}(\tau, \bar{\tau}) = &  \psi^{\vee, \eta}_{p,a}(\tau) + G_{p,a,c_{1}}(\tau, \bar{\tau} ) -G_{p,a,c_{2}}(\tau, \bar{\tau}),
\end{align}
where 
\begin{align}
    G_{p,a,c_{1}}(\tau, \bar{\tau} ) = \frac{-i  e^{-\frac{i \pi }{6}} }{2\eta(\tau) \sqrt{2p} } \sum_{s \in \{0,2\}} e^{\frac{i \pi s}{6}} \sum_{\ell \in \frac{2 s +1}{6} + \Z  } e^{i \pi \ell } q^{\frac{3 \ell^{2}}{2}}   \int_{- \bar{\tau}}^{\infty }  \frac{\theta_{p,a}^{1}(y)}{\sqrt{-i(y + \tau)}} \dd y, 
\end{align}
and 
\begin{align}
     G_{p,a,c_{2}}(\tau, \bar{\tau} ) = &  \frac{-i  e^{-\frac{i \pi }{6}} }{2\eta(\tau) \sqrt{2p} } \sum_{s \in \{0,2\}} e^{\frac{i \pi s}{6}} \sum_{\ell \in \frac{2 s +1}{6}u + \Z  } e^{i \pi \ell u} q^{\frac{3}{2}\ell^{2}}  e^{ - \pi i v ( a u - (2 s +1)p v )} \int_{- \bar{\tau}}^{\infty }  \frac{\theta_{p,a u -(2s+1)p v}^{1}(y)}{\sqrt{-i(y + \tau)}} \dd y. 
 \end{align}
Using that $v$ is even and $\eta(\tau) = \sum_{n \in \Z} e^{i \pi n} q^{\frac{3}{2} (n - \frac{1}{6})^{2} }$, these simplify to 
\begin{align}
    G_{p,a,c_{1}}(\tau, \bar{\tau} ) =& \frac{-i   }{ \sqrt{2p} } \int_{- \bar{\tau}}^{\infty }  \frac{\theta_{p,a}^{1}(y)}{\sqrt{-i(y + \tau)}} \dd y, \endline
         G_{p,a,c_{2}}(\tau, \bar{\tau} ) = &  \frac{-i   }{ \sqrt{2p} }\left(\frac{12}{u}\right) \int_{- \bar{\tau}}^{\infty }  \frac{\theta_{p,a u }^{1}(y)}{\sqrt{-i(y + \tau)}} \dd y. 
\end{align}
Since $\left(\frac{12}{u}\right) u = -1 \mod 2p$ and $\theta_{p,a}^{1}=- \theta_{p,-a}^{1} $, we obtain 
\be 
G_{p,a,c_{1}}(\tau, \bar{\tau}) - G_{p,a,c_{2}}(\tau, \bar{\tau}) =   \sqrt{\frac{-2i}{p}} \int_{- \bar{\tau}}^{\infty }  \frac{\theta_{p,a}^{1}(y)}{\sqrt{(y + \tau)}} \dd y.
\ee 
Noting that $\overline{\theta_{p,a}^{1}(- \overline{y})} = \theta_{p,a}^{1}(y)$ since $\theta_{p,a}^{1}(y)$ has real coefficients, we see that $\hat{\psi}^{\vee, \eta}_{p,a}$ has shadow $-\frac{2}{\sqrt{p}}\theta^{1}_{p,a}$.

The $T$ transformation of $\hat{\psi}^{\vee, \eta}_{p,a}$ follows from those of Zwegers' theta function and the Dedekind eta function:
$$\hat{\psi}^{\vee, \eta}_{p,a}(\tau +1, \bar{\tau} +1) = e^{- 2 \pi i\frac{a^{2}}{4p}} \hat{\psi}^{\vee, \eta}_{p,a}(\tau, \bar{\tau}).$$
 
For the $S$ transformation, we can combine the transformation of the modular completion of the indefinite theta functions (see Appendix A of \cite{Harichurn:2025suf}, and Corollary 2.9 of \cite{zwegers2008mock}) with that of the Dedekind eta function:
\begin{align}
   & \hat{\psi}^{\vee, \eta}_{p,a}(\tau, \bar{\tau}) \endline
   = &\frac{e^{- \frac{i \pi }{6}}}{2\eta(\tau)} \sum_{s \in \{0,2\}} e^{\frac{i \pi  s}{6}} \frac{e^{\frac{i \pi (2s +1)}{6}} (-\tau)^{-1}}{\sqrt{6p}} \sum_{ b =1}^{2p} \sum_{r = 0}^{2} e^{\frac{i \pi a b}{p} } e^{-\frac{i \pi (2s +1)(2r+1)}{6}} e^{- \frac{i \pi (2r+1)}{6}} \hat{\theta}_{A,(\frac{b}{2p}, \frac{2r+1}{6}), \beta, c_{1}, c_{2}} (\tfrac{-1}{\tau}, \tfrac{-1}{\bar{\tau}}) \endline
   = &\frac{e^{- \frac{i \pi }{6}}}{2\eta(\tfrac{-1}{\tau})}  \frac{\sqrt{-i \tau}(-\tau)^{-1}}{\sqrt{2p}} \sum_{ b =1}^{2p} \sum_{r \in \{0,2\}} e^{\frac{i \pi a b}{p} } e^{\frac{i \pi r}{6}} \hat{\theta}_{A,(\frac{b}{2p}, \frac{2r+1}{6}), \beta, c_{1}, c_{2}} (\tfrac{-1}{\tau}, \tfrac{-1}{\bar{\tau}})\endline
   = & \frac{ 1}{\sqrt{-i \tau}} \sum_{ b =1}^{2p} \sqrt{\frac{2}{p}} \frac{i e^{\frac{i \pi a b}{p} }}{2} \hat{\psi}_{p,b}^{\vee, \eta}(\tfrac{-1}{\tau}, \tfrac{-1}{\bar{\tau}}),
\end{align}
where $\hat{\theta}_{A,(\frac{b}{2p}, \frac{2r+1}{6}), \beta, c_{1}, c_{2}}$ is the modular completion of $\theta_{A,(\frac{b}{2p}, \frac{2r+1}{6}), \beta, c_{1}, c_{2}}$. By Corollary 2.9 from \cite{zwegers2008mock}, $\psi_{p,a}^{\vee, \eta} =- \psi_{p,-a}^{\vee, \eta}  $, and hence,
\be 
\hat{\psi}^{\vee, \eta}_{p,a}(\tau, \bar{\tau})  = -(-i \tau)^{-\frac{1}{2}} \sum_{b =1}^{p-1} \sqrt{\frac{2}{p}} \sin(\frac{\pi a b }{p}) \hat{\psi}_{p,b}^{\vee, \eta}(\tfrac{-1}{\tau}, \tfrac{-1}{\bar{\tau}}).
\ee 
Since $S$ and $T$ generate $\mathrm{Mp}(2,\Z)$, $\hat{\psi}^{\vee, \eta}_{p,a}$ is a harmonic Maass form of weight $\frac{1}{2}$ with respect to $\rho_{-p}$. Hence $ \psi^{\vee, \eta}_{p,a}$ is a vector-valued mock modular form of weight $\frac{1}{2}$ with respect to $\rho_{-p}$, with shadow $-\frac{2}{\sqrt{p}}\theta^{1}_{p,a}$.      
\end{proof}

\noindent The definition \eqref{sparthdualfalsedef} can be brought to a simpler form. Writing out the indefinite theta functions, we have 
\begin{align}
& \psi^{\vee, \eta}_{p,a}(\tau) \endline
     = & \frac{1}{2\eta(\tau)} \sum_{ n_{1} \in \frac{a}{2p} + \Z }  \left( e^{- \tfrac{i \pi }{6}} \sum_{n_{2} \in \frac{1}{6} + \Z } + e^{ \tfrac{i \pi }{6}} \sum_{n_{2} \in \frac{5}{6} + \Z }  \right)  (\mathrm{Sign}(  u n_{1} - 3 v n_{2} )-\mathrm{Sign}(n_{1}) ) e^{i \pi n_{2} } q^{- p n_{1}^{2}  + \frac{3}{2} n_{2}^{2} }   \endline
    = & \frac{1}{2\eta(\tau)} \sum_{x \in a + 2p \Z } \sum_{y \in 1+ 2 \Z} (\mathrm{Sign}(u x -p v y) - \mathrm{Sign}(x) ) \left( \frac{12}{y} \right)q^{\frac{y^{2}}{24 } - \frac{x^{2}  }{4p} } \endline
    = & \frac{-1}{\eta(\tau)} \sum_{\substack{x \in a + 2 p \Z \\ y \in 1+2 \Z \\ 0 < \frac{x}{y} < \frac{p v}{u} }} \mathrm{Sign}(x) \left( \frac{12}{y} \right) q^{\frac{y^{2}}{24 } - \frac{x^{2}  }{4p} },
\end{align}
where in the second line we substituted $x = 2 p n_{1}$, and $y = 6 n_{2}$. The factor $ (\mathrm{Sign}(u x -p v y) - \mathrm{Sign}(x) )$ is non-zero if and only if $0 < \frac{x}{y} < \frac{p v}{u}  $, in which case it equals $- 2 \mathrm{Sign}(x)$. We therefore arrive at the simple form
\be \label{sparthdualfalsedef2}
\psi_{p,a}^{\vee, \eta}(\tau) = \frac{-1}{\eta(\tau)} \sum_{\substack{x \in a + 2 p \Z \\ y \in 1+2 \Z \\ 0 < \frac{x}{y} < \frac{p v}{u} }} \mathrm{Sign}(x) \left( \frac{12}{y} \right) q^{\frac{y^{2}}{24 } - \frac{x^{2}  }{4p} }.
\ee 
In particular, $\psi_{p,a}^{\vee, \eta}(\tau)$ has integer coefficients. 

%%%%%%%%%%%%%%%%%%%%%%%%%%%%%%%%%%%%%%%

\subsection{Examples of Dual False Theta Functions}\label{examplesofdualfalse}

In this subsection, we will write down the first few terms of dual false theta functions from the Li--Schwagenscheidt family and the special arithmetic family. We have multiplied the $q$-series $\psi_{p,a}^{\vee}$ with $\mathcal{N}_{a}$ so that they have integer coefficients. Using these first few terms, we compare the $q$-series with those in the literature.

\noindent \textbf{Li--Schwagenscheidt family} \par \nobreak

\begin{itemize}
    \item $p=3$ 
\begin{align*}
    \psi_{3,1}^{\vee, \eta^{3}}(\tau) =&   q^{-\frac{1}{12}}-16 q^{11/12}-55 q^{23/12}-144 q^{35/12}-330 q^{47/12}-704 q^{59/12} + \cdots \\
        \psi_{3,2}^{\vee, \eta^{3}}(\tau) =&-10 q^{2/3}-44 q^{5/3}-110 q^{8/3}-280 q^{11/3}-572 q^{14/3}-1200 q^{17/3} + \cdots  
\end{align*}
The coefficients of $\psi_{3,1}^{\vee, \eta^{3}}$ match the coefficients of the mock modular form $H_{1}^{(3)}$ from \cite{Cheng:2012tq}.
    \item $p=5$ 
    \begin{align*}
        \psi_{5,1}^{\vee, \eta^{3}}(\tau) =&  q^{-\frac{1}{20}} -4 q^{19/20}-9 q^{39/20}-20 q^{59/20}-35 q^{79/20}-60 q^{99/20} + \cdots  \\
        \psi_{5,2}^{\vee, \eta^{3}}(\tau) =& -5 q^{4/5}-15 q^{9/5}-26 q^{14/5}-54 q^{19/5}-90 q^{24/5} + \cdots  \\ 
        \psi_{5,3}^{\vee, \eta^{3}}(\tau) =& -4 q^{11/20}-11 q^{31/20}-24 q^{51/20}-45 q^{71/20}-80 q^{91/20} + \cdots  \\ 
        \psi_{5,4}^{\vee, \eta^{3}}(\tau) =& - q^{\frac{1}{5}} -6 q^{6/5}-10 q^{11/5}-25 q^{16/5}-36 q^{21/5}+ \cdots  
    \end{align*}
   The coefficients of $\psi_{5,1}^{\vee, \eta^{3}}$ match the coefficients of the mock modular form $H_{1}^{(5)}$ from \cite{Cheng:2012tq}. 
   \item $p=7$ 
   \begin{align*}
       \psi_{7,1}^{\vee, \eta^{3}}(\tau) = &q^{- \frac{1}{28}} -2 q^{27/28}-3 q^{55/28}-5 q^{83/28}-10 q^{111/28}-15 q^{139/28} + \cdots  \\
       \psi_{7,2}^{\vee, \eta^{3}}(\tau) = & -2 q^{6/7}-6 q^{13/7}-10 q^{20/7}-16 q^{27/7}-24 q^{34/7}-40 q^{41/7} + \cdots  \\
       \psi_{7,3}^{\vee, \eta^{3}}(\tau) = &-3 q^{19/28}-6 q^{47/28}-11 q^{75/28}-18 q^{103/28}-29 q^{131/28}-45 \
q^{159/28}  + \cdots  \\
       \psi_{7,4}^{\vee, \eta^{3}}(\tau) = & -2 q^{3/7}-6 q^{10/7}-8 q^{17/7}-18 q^{24/7}-24 q^{31/7}-42 q^{38/7}  + \cdots  \\
       \psi_{7,5}^{\vee, \eta^{3}}(\tau) = & -q^{3/28}-3 q^{31/28}-7 q^{59/28}-11 q^{87/28}-18 q^{115/28}-28 q^{143/28} + \cdots  \\
       \psi_{7,6}^{\vee, \eta^{3}}(\tau) = & -2 q^{5/7}-2 q^{12/7}-6 q^{19/7}-6 q^{26/7}-16 q^{33/7} + \cdots  \\
   \end{align*}
   \item $p=13$
   \begin{align*}
    \psi_{13,1}^{\vee, \eta^{3}}(\tau) = &  3 q^{-\tfrac{1}{52}}-5 q^{51/52}-15 q^{103/52}-32 q^{155/52}-71 q^{207/52}-141 q^{259/52}+ \cdots \\
    \psi_{13,2}^{\vee, \eta^{3}}(\tau) = & q^{- \tfrac{1}{13}} -8 q^{12/13}-24 q^{25/13}-56 q^{38/13}-121 q^{51/13}-243 q^{64/13}+ \cdots
   \end{align*}
\end{itemize}

The $\hat{Z}^{\eta^{3}}$ invariants obtained from the Li--Schwagenscheidt family are as follows:
\begin{itemize}
    \item $\hat{Z}^{\eta^{3}}(\overline{\Sigma(2,3,5)},q) = \chi_{0}(q)-2$, where $\chi_{0}(q)$ is Ramanujan's order $5$ mock theta function. Here the line over $\Sigma(2,3,5)$ denotes orientation reversal. This matches the $\hat{Z}$ obtained from the regularized surgery formula \cite{Park:2021ufu, Cheng:2024vou} and resurgence techniques \cite{Adams:2025aad}.
    \item $\hat{Z}^{\eta^{3}}(\overline{\Sigma(2,3,7)},q) = F_{0}(q)$, where $F_{0}(q)$ is Ramanujan's order $7$ mock theta function. This matches the $\hat{Z}$ obtained from the regularized surgery formula \cite{Park:2021ufu, Cheng:2024vou} and resurgence techniques \cite{Adams:2025aad}.
    \item $\hat{Z}^{\eta^{3}}(\overline{\Sigma(2,3,11)},q) = M_{11,1}(\tau) $, where $M_{11,1}(q)$ is Zagier's order $11$ mock theta function. This matches the $\hat{Z}$ obtained from resurgence techniques \cite{Adams:2025aad}. 
    \item $\overline{\Sigma(2,3,13)}$: 
    \begin{align*}
        \hat{Z}^{\eta^{3}}(\overline{\Sigma(2,3,13)},q) = &  25 q^{263/312}+102 q^{575/312}+255 q^{887/312}+615 q^{1199/312}+1230 \
q^{1511/312} \\ & +2490 q^{1823/312}+4597 q^{2135/312}+8355 \
q^{2447/312}+ \cdots 
    \end{align*}
    This $q$-series does not match $\hat{Z}$ obtained from the regularized surgery formula \cite{Park:2021ufu, Cheng:2024vou}, or from resurgence techniques \cite{Adams:2025aad}.
\end{itemize}

\noindent \textbf{Special Arithmetic family} \par \nobreak

\begin{itemize}
    \item The first few coefficients of $\psi_{7,a}^{\vee, \eta}$ and $\psi_{7,a}^{\vee, \eta^{3}}$ match.
    \item $p=13$
    \begin{align*}
        \psi_{13,1}^{\vee, \eta}(\tau) = & 2 q^{- \frac{1}{52}} -2 q^{51/52} -2 q^{103/52} -2q^{207/52}  + \cdots \\
        \psi_{13,2}^{\vee, \eta}(\tau) = & -2 q^{25/13} -4 q^{38/13} -4q^{51/13} -6 q^{64/13} - 8 q^{77/13} + \cdots 
    \end{align*}
    Note that these $q$-series differ from $\psi_{13,a}^{\vee, \eta^{3}}$.
    \item We can compute the $\hat{Z}^{\eta}$ for the Brieskorn sphere $\overline{\Sigma(5,7,11)}$,
    \begin{align*}
        \hat{Z}^{\eta}(\overline{\Sigma(5,7,11)},q) = & q^{54/385}+q^{439/385}-q^{1209/385}-q^{1594/385}-5 q^{1979/385} \\ &  -8 q^{2364/385}-12 q^{2749/385}-20 q^{3134/385} + \cdots
    \end{align*}
\end{itemize}

%%%%%%%%%%%%%%%%%%%%%%%%%%%%%%%%%%%%%%%

\subsection{Non-uniqueness of Dual False Theta Functions}\label{nonuniquenessdiscussio}
The examples in Subsection \ref{examplesofdualfalse} show that the characterization of dual false theta functions does not determine them uniquely. For $p=13$, the Li--Schwagenscheidt and the special arithmetic families give different $q$-series. This non-uniqueness is built into the characterization, as the commutative diagram \eqref{commutativediagmocktheta} shows. Given a dual false theta function $\psi_{p, a}^{\vee}$, we may add to it a vector-valued weakly holomorphic modular form of weight $\frac{1}{2}$ with respect to $\rho_{-p}$ and obtain another dual false theta function. Recall that weakly holomorphic forms can have singularities of the type $q^{-c}$ at $\Q \cup \{\infty \}$. The diagram also tells us that the dimension of the affine space of vector-valued mock modular forms of weight $\frac{1}{2}$ with respect to $\rho_{-p}$ and with a fixed shadow equals $\mathrm{dim}(M_{\frac{1}{2},-p}^{!})$. The vector space $M_{\frac{1}{2},-p}^{!}$ of weakly holomorphic modular forms of weight $\frac{1}{2}$ with respect to $\rho_{-p}$ is infinite dimensional; however, once we restrict the allowed pole order at $\Q \cup \{\infty\}$, it becomes finite-dimensional \cite{zagier2009ramanujan}.  

This raises the question of whether a finer characterization singles out $\hat{Z}$ uniquely, or whether $\hat{Z}$ is inherently non-unique. Below we discuss three possible ways to single out a unique dual false theta function $\psi_{p, a}^{\vee}$ or $\hat{Z}$:
\begin{itemize}
    \item \textbf{$c_{\text{eff}}\,$-minimization}: The effective central charge quantifies the growth rate of coefficients of supersymmetric partition functions of $3d$ $\mathcal{N}=2$ theories \cite{Gukov:2023cog}. We can extend the notion of $c_{\text{eff}}$ to $q$-series as in \cite{Harichurn:2025suf}. The $c_{\text{eff}}$ of a mock theta function is given by 
    \be 
    c_{\text{eff}} = 24 \, \mathrm{max}\{ c \vert q^{-c} \text{ is a pole at $\Q \cup \{\infty\}$} \}.
    \ee 
    Since restricting the pole order $q^{-c}$ renders the space of possible dual false theta functions finite-dimensional, we could single out a unique one by choosing the representative that minimizes $c_{\text{eff}}$. While it is not clear in general whether a unique mock theta function minimizes $c_{\text{eff}}$, \cite{Costin:2026bdd} established this for Ramanujan's order $5$ mock theta function, and \cite{Costin:2026new} establishes it for Ramanujan's order $3$ pair $(f, \omega)$ and the mock theta functions of orders $7$ and $11$. 

    Note that the dual false theta functions in the Li--Schwagenscheidt family are of the form $(q\text{-series with positive exponents})/ \eta^{3}(\tau)$. Consequently, $\psi_{p,a}^{\vee, \eta^{3}}$ cannot have a pole larger than $q^{- \frac{1}{8}}$, so 
    $$ c_{\text{eff}}(\psi_{p,a}^{\vee, \eta^{3}} ) \leq 3 .$$ 
    Similarly, for the dual false theta functions in the special arithmetic family, 
    $$ c_{\text{eff}}(\psi_{p,a}^{\vee, \eta} ) \leq 1 .$$
    
    \item \textbf{Denominators of coefficients}: The coefficients of the dual false theta functions with shadow $\tfrac{1}{\sqrt{p}}\theta_{p, a}^{1}$ obtained in this article are rational numbers with bounded denominators. One could therefore single out a dual false theta function by requiring coefficients with minimal denominators, or with denominator 1. Since the coefficients of dual false theta functions from the special arithmetic family are all integers, they satisfy the denominator $1$ characterization. 
    
    \item \textbf{Neumann invariance}: Extending the techniques used to construct the $\hat{Z}$ invariant for Seifert manifolds to general plumbed manifolds remains an open problem. Any framework that computes the $\hat{Z}$ invariant for positive-definite plumbed manifolds should be Neumann invariant. That is, plumbing graphs related by Neumann moves \cite{Neu} should yield the same $q$-series. This topological constraint could serve to single out a dual false theta function.
\end{itemize}
These criteria are not mutually exclusive, and it would be interesting to understand whether they select the same dual false theta function. We leave these questions for future work.

%%%%%%%%%%%%%%%%%%%%%%%%%%%%%%%%%%%%%%%

\section*{Acknowledgments}
We thank Miranda Cheng, John Chae, Gerald Dunne, Pedro Guicardi, Sergei Gukov, Oguz Onner, Paul Orland, Davide Passaro, Lara San Martin Suarez and Josef Svoboda for very helpful discussions. While writing this manuscript, we learned that \cite{Costin:2026new} was soon to be published, and we are especially grateful to Ovidiu Costin, Gerald Dunne, and Ali Saraeb for sharing their draft and agreeing to coordinate publications. This work is supported by the Walter Burke Institute for Theoretical Physics, the U.S. Department of Energy, Office of Science, Office of High Energy Physics, under Award No. DE-SC0011632 and the Simons Collaboration Grant on ``New Structures in Low-Dimensional Topology''. We acknowledge the use of Claude Opus 5.5 to generate the Mathematica script used in Section \ref{examplesofdualfalse} for computing the $q$-series expansion. All formulas and results are our own, and we verified the code logic and the accuracy of the resulting series.

\printbibliography[heading=bibintoc]

\end{document}